\pdfoutput=1
\documentclass[12pt]{article}

\usepackage[margin=1.15in]{geometry}
\usepackage[T1]{fontenc}
\usepackage{microtype}
\usepackage{amssymb}
\usepackage{amsthm}
\usepackage{enumitem}
\usepackage{mathdots}
\usepackage{mathtools}
\usepackage{mleftright}
\usepackage[numbers]{natbib}
\usepackage{nth}
\usepackage{xcolor}

\usepackage{bussproofs}
\EnableBpAbbreviations

\usepackage{tikz}
\usetikzlibrary{calc}
\tikzset{
  horizontal center at origin/.style={
    execute at end picture={
      \path let \p1=(current bounding box.west),\p2=(current bounding box.east)
      in ({-max(-1*\x1,\x2)},\y1) ({max(-1*\x1,\x2)},\y1);
    }
  },
  dots/.style args={#1 per #2}{
    line cap=round,
    dash pattern=on 0 off #2/#1
  }
}

\usepackage{hyperref}
\colorlet{darkblue}{blue!75!black}
\hypersetup{
  colorlinks,
  linkcolor=darkblue,
  citecolor=darkblue,
  urlcolor=darkblue
}

\usepackage[nameinlink,capitalize]{cleveref}
\crefname{equation}{}{}
\crefname{figure}{Figure}{Figures}

\theoremstyle{plain}
\newtheorem{theorem}{Theorem}[section]

\newtheorem{proposition}[theorem]{Proposition}
\newtheorem{corollary}[theorem]{Corollary}

\theoremstyle{definition}
\newtheorem{definition}[theorem]{Definition}
\newtheorem{remark}[theorem]{Remark}
\newtheorem{question}[theorem]{Question}

\newcommand{\NN}{\mathbb{N}}
\newcommand{\bL}{\mathbf{L}}
\newcommand{\bV}{\mathbf{V}}
\newcommand{\cC}{\mathcal{C}}
\DeclarePairedDelimiter{\sizeOf}{\lvert}{\rvert}
\DeclareMathOperator{\poly}{poly}
\DeclareMathOperator{\var}{var}

\newcommand{\false}{\mathit{False}}
\newcommand{\true}{\mathit{True}}
\newcommand{\liff}{\leftrightarrow}
\newcommand{\lneg}[1]{\overline{#1}}

\newcommand{\complexityclass}[1]{\mathsf{#1}}
\newcommand{\PTIME}{\complexityclass{P}}
\newcommand{\NP}{\complexityclass{NP}}
\newcommand{\proofsystem}[1]{\mathrm{#1}}
\newcommand{\CP}{\proofsystem{CP}}

\title{Regular~resolution effectively~simulates resolution}
\author{
  Sam Buss\thanks{
    Department of Mathematics, University of California, San Diego.
    Email: \texttt{sbuss@ucsd.edu}.
  }
  \and
  Emre Yolcu\thanks{
    Computer Science Department, Carnegie Mellon University.
    Email: \texttt{eyolcu@cs.cmu.edu}.
  }
}

\begin{document}

\maketitle

\begin{abstract}
  Regular resolution is a refinement of the resolution proof system
  requiring that no variable be resolved on
  more than once along any path in the proof.
  It is known that there exist sequences of formulas
  that require exponential-size proofs in regular resolution
  while admitting polynomial-size proofs in resolution.
  Thus, with respect to the usual notion of simulation,
  regular resolution is separated from resolution.
  An alternative, and weaker, notion for comparing proof systems
  is that of an ``effective simulation,''
  which allows the translation of the formula
  along with the proof when moving between proof systems.
  We prove that regular resolution is
  equivalent to resolution under effective simulations.
  As a corollary, we recover in a black-box fashion
  a recent result on the hardness of automating regular resolution.
\end{abstract}

\noindent
{\small Keywords: \emph{resolution, regular resolution, effective simulation,
    automatability, proof complexity}}

\section{Introduction}
\label{sec:introduction}

Proof complexity studies the sizes of proofs%
\footnote{Throughout this work, by ``proof''
  we mean a refutation of satisfiability.}
in propositional proof systems.
A common question in proof complexity is
that of the relative strengths of different systems.
The comparison is performed typically with respect to
the notion of a ``simulation''~\cite[Definition~1.5]{CR79}.
A system~$P$ \emph{simulates} another system~$Q$
if every~$Q$-proof can be converted,
with at most a polynomial increase in size,
into a~$P$-proof of the same formula.
An alternative, and weaker, notion of simulation is the following,
which is arguably more natural from an algorithmic point of view.
(See Pitassi and Santhanam~\cite[Section~1]{PS10} for a discussion.)

\begin{definition}[{\cite[Definition~2.5]{PS10}}]%
  \label{def:effective-simulation}
  Let $P$~and~$Q$ be two proof systems
  for a class~$\cC$ of propositional formulas.
  The \emph{size}, $\sizeOf{\Gamma}$, of a formula~$\Gamma$ is
  defined to equal the number of symbols in the formula.
  We say $P$ \emph{effectively simulates}~$Q$ if there exists
  a function~$f \colon \cC \times \NN \to \cC$ such that the following hold.
  \begin{itemize}
  \item The formula~$f(\Gamma, s)$
    can be computed in time polynomial in~$\sizeOf{\Gamma} + s$,
    and it is satisfiable if and only if $\Gamma$~is.
  \item When $s$ is at least
    the size of the smallest~$Q$-proof of~$\Gamma$,
    the formula~$f(\Gamma, s)$ has a~$P$-proof
    of size polynomial in~$\sizeOf{\Gamma} + s$.
  \end{itemize}
\end{definition}

\begin{remark}\label{rem:effective-simulation-parameter}
  As remarked by Pitassi and Santhanam~\cite{PS10},
  the role of the size parameter~$s$ in the above definition
  might not be clear at first glance.
  It would be simpler to define a notion of
  \emph{strict effective simulation} by omitting~$s$
  and requiring that $f(\Gamma)$~be computable
  in time polynomial in~$\sizeOf{\Gamma}$
  and that $f(\Gamma)$ have a~$P$-proof of size polynomial
  in the size of the smallest~$Q$-proof of~$\Gamma$.
  A major motivation behind \cref{def:effective-simulation}
  is its relationship to ``weak automatability''~\cite[Definition~4]{AB04},
  which we define in \cref{sec:automatability},
  and the relaxed definition suffices for the relationship to hold.
\end{remark}

Effective simulations exist in several instances
where either no simulation is known or a separation exists.
Examples include the following,%
\footnote{Some results use a version of \cref{def:effective-simulation}
  that allows comparing proof systems
  over different languages~\cite[Definition~4]{HHU07}.}
where $P \geq Q$ denotes that $P$ effectively simulates~$Q$:
\begin{itemize}[itemsep=0pt]
\item $\text{linear resolution} \geq \text{resolution}$~\cite{BP07,BJ16}
\item $\text{clause learning} \geq \text{resolution}$~\cite{HBPV08,BHJ08}
\item $\text{resolution} \geq \text{$k$-DNF resolution}$~\cite{AB04}
\item $\text{blocked clauses without new variables}
  \geq \text{extended resolution}$~\cite{BT21}
\item $\proofsystem{G}_0 \text{ (``quantified Frege'')}
  \geq \text{any quantified propositional proof system}$~\cite{PS10}
\item $\CP \text{ (cutting planes)}
  \geq \CP \text{ with quadratic terms}$~\cite{Pud03}
\item $\text{constant-depth polynomial calculus}
  \geq \CP,\, \text{Positivstellensatz calculus}$~\cite{IMP20}
\end{itemize}

In this work we prove that
although regular resolution and resolution are separated
with respect to the usual notion of simulation~\cite{Goe93,AJPU07,Urq11},
the two systems are equivalent under effective simulations.
As a technical side note,
although in most of the known effective simulations
the size parameter~$s$ is not needed
(i.e., they are strict effective simulations),
in our simulation it is necessary
that $f$~have access to the parameter.
It is an open question whether regular resolution
strictly effectively simulates resolution.

\section{Preliminaries}
\label{sec:preliminaries}

We assume that the reader is familiar with
propositional logic, proof complexity, and resolution.
We review some concepts to describe our notation.
For all~$n \in \NN$, we let $[n] \coloneqq \{1, \dots, n\}$.
For~$X$ a~nonempty list of variables, we write $\poly(X)$
to denote a quantity bounded by some polynomial in~$X$.

A \emph{literal} is a propositional variable~$x$ or its negation~$\lneg{x}$.
Overline denotes negation,
and if $p$ is the literal~$\lneg{x}$, then $\lneg{p}$ is~$x$.
A \emph{clause} is a disjunction of literals.
We use~$\bot$ to denote the empty clause.
A formula in \emph{conjunctive normal form} (CNF) is a conjunction of clauses.
Throughout this work, by ``formula'' we mean a formula in~CNF\@.
We identify clauses with sets of literals and formulas with sets of clauses.
For a formula~$\Gamma$, we denote by~$\var(\Gamma)$
the set of all the variables occurring in~$\Gamma$.
In particular, for a literal~$p$ of a variable~$x$, we have $\var(p) = x$.

\begin{definition}\label{def:resolution-rule}
  The \emph{resolution rule} is
  \begin{center}
    \AXC{$A \lor x$}
    \AXC{$B \lor \lneg{x}$}
    \BIC{$A \lor B$}
    \DisplayProof,
  \end{center}
  where $A$,~$B$ are clauses
  and $x$ is a variable not occurring in $A$~or~$B$.
  We call $A \lor B$ the \emph{resolvent}
  of~$A \lor x$ and~$B \lor \lneg{x}$ on~$x$.
\end{definition}

\begin{definition}\label{def:resolution-derivation}
  A \emph{resolution derivation} from a formula~$\Gamma$
  is a sequence~$\Pi = C_1, \dots, C_s$ of distinct clauses such that
  for all~$i \in [s]$, the clause~$C_i$ either occurs in~$\Gamma$
  or is a resolvent of two earlier clauses in the sequence.
  If $C_s = \bot$, then $\Pi$ is a \emph{resolution refutation} of~$\Gamma$.
  The \emph{size} of~$\Pi$ is~$s$.
\end{definition}

\begin{remark}\label{rem:resolution-graph}
  A resolution derivation~$\Pi$ from~$\Gamma$
  can be viewed as a directed acyclic graph:
  The nodes of the graph are the clauses in~$\Pi$;
  every \emph{initial clause} (i.e., a clause in~$\Gamma$) has in-degree zero,
  and every other clause~$D$ has two incoming edges
  from the premises of the resolution inference that derives~$D$.
  Thus, every node has in-degree zero or two.
  The \emph{height} of a resolution derivation
  is the number of edges in the longest directed path in the graph.
\end{remark}

\begin{definition}
  A resolution derivation~$\Pi$ is \emph{regular}
  if no variable is resolved on more than once along any directed path in~$\Pi$.
\end{definition}

In the rest of this paper, by ``path'' we mean a directed path.

\section{Main result}
\label{sec:main-result}

\begin{theorem}\label{thm:regular-effectively-general}
  Regular resolution effectively simulates resolution.
\end{theorem}
\begin{proof}
  Let $\Gamma$ be a formula with $n$~variables,
  and let $h$ be a size parameter.
  We will define a new formula~$f(\Gamma, h)$ such that
  if $\Gamma$ has a resolution refutation of size~$s$ and height~$h$,
  then $f(\Gamma, h)$ has a regular resolution refutation
  of size~$\poly(s, n)$ and height~$\poly(h, n)$.

  The formula~$f(\Gamma, h)$ is defined by introducing new variables
  and adding several $2$-clauses to~$\Gamma$.
  For each variable~$x$ of~$\Gamma$ and each~$j \in [h - 1]$,
  there is a new variable~$W[x, j]$.
  We refer to~$j$ as the \emph{level} of~$W[x, j]$.
  We identify~$W[x, h]$ with~$x$.
  Thus, $f(\Gamma, h)$ has a total of $h n$~variables,
  with $(h - 1) n$~of~them new.
  We extend the notation to define~$W[p, j]$ for~$p$ a~literal
  by letting $W[\lneg{x}, j]$ denote the literal~$\lneg{W[x, j]}$.

  We form~$f(\Gamma, h)$ by adding to~$\Gamma$ the $2$-clauses expressing
  for all~$j \in [h - 1]$ the equivalence~$W[x, j] \liff W[x, j + 1]$.
  That is,
  \begin{equation}\label{eq:addedTwoClauses}
    f(\Gamma, h) \coloneqq \Gamma
    \land \bigwedge_{x \in \var(\Gamma)} \bigwedge_{j \in [h - 1]}
    \mleft[
    \left(\lneg{W[x, j]} \lor W[x, j + 1]\right) \land
    \left(W[x, j] \lor \lneg{W[x, j + 1]}\right)
    \mright].
  \end{equation}
  The formula~$f(\Gamma, h)$ is satisfiable
  if and only if $\Gamma$ is satisfiable since the added clauses
  simply define new names for each variable of~$\Gamma$.

  Let $\Pi = C_1, \dots, C_s$ be
  a size-$s$, height-$h$ resolution refutation of~$\Gamma$,
  viewed as a directed graph as described in \cref{rem:resolution-graph}.
  To prove the theorem, we will describe how to turn~$\Pi$ into
  a regular resolution refutation~$\Pi'$ of~$f(\Gamma, h)$
  of size at most~$6hns$ and height at most~$hn$.
  The intuition for forming the regular refutation
  is that the new variables~$W[x, j]$ are equivalent to~$x$,
  and the general refutation~$\Pi$ can be turned into
  a regular refutation by replacing literals~$p$ with~$W[p, j]$,
  letting $j$ decrease as the refutation progresses.
  In this way, multiple resolutions on a variable~$x$
  are replaced by resolutions on variables~$W[x, j]$,
  with $j$ decreasing along paths in the refutation
  so that no~$W[x, j]$ is resolved on twice on any path.

  Let $D$ be a clause in~$\Pi$.
  Without loss of generality, there is at least
  one path in~$\Pi$ from~$D$ to~$\bot$.
  For a variable~$x$, the \emph{irregularity height}
  of~$x$ at~$D$ in~$\Pi$, denoted~$L_D(x)$,
  is defined to be the maximum number of inferences
  that use~$x$ as the resolution variable
  along any path in~$\Pi$ from~$D$ to~$\bot$.
  For~$p$ a~literal, we allow~$L_D$ to act on~$p$
  by letting $L_D(p) \coloneqq L_D(\var(p))$.
  The value of~$L_D(x)$ depends on~$\Pi$ of course,
  but this is suppressed in the notation.
  Note that $L_D(x) \leq h$.

  For each clause~$D$ in~$\Pi$, let
  \begin{equation*}
    g(D) \coloneqq \bigvee_{p \in D} W[p, L_D(p)].
  \end{equation*}
  The regular resolution refutation~$\Pi'$
  will have the form $P_0, P_1, \dots, P_s$,
  where each~$P_i$ is a finite sequence of clauses.
  The sequence~$P_0$ contains
  the initial clauses from~$\Gamma$ that appear in~$\Pi$
  and the newly added $2$-clauses in~$f(\Gamma, h)$.
  For all~$i \in [s]$, the sequence~$P_i$
  will end with the clause~$g(C_i)$.
  Since $g(\bot) = \bot$, the final clause of~$\Pi'$ will be~$\bot$,
  so $\Pi'$ will be a regular resolution refutation of~$f(\Gamma, h)$.

  The principal tool in forming each~$P_i$ will be
  ``lowering'' the levels of literals.
  Specifically, suppose that $E$~and~$F$ are clauses of the forms
  \begin{equation*}
    E = W[p_1, j_1] \lor \dots \lor W[p_t, j_t]
    \quad \text{ and } \quad
    F = W[p_1, k_1] \lor \dots \lor W[p_t, k_t]
  \end{equation*}
  with~$j_\ell \geq k_\ell$ for all~$\ell \in [t]$.
  When this holds, we say $E$ \emph{dominates}~$F$.
  Then $F$ can be derived from~$E$ by resolving it
  with the new~$f(\Gamma, h)$ clauses
  \begin{equation*}
    W[p_\ell, m] \lor \lneg{W[p_\ell, m + 1]} \qquad
    \text{for~$\ell \in [t]$
      and $m = j_\ell - 1, j_\ell - 2, \dots, k_\ell$.}
  \end{equation*}
  This is called \emph{lowering} $E$ to~$F$.

  We now describe how to inductively form
  the subderivations~$P_i$ of~$\Pi'$.
  Consider the clause~$C_i$,
  which we will write simply as~$C$ from this point on.
  The subderivation~$P_i$ needs to end with~$g(C)$.
  There are two cases to consider.
  \begin{description}
  \item[Case 1:] \textit{$C$ is an initial clause in~$\Pi$.}
    Since $C \in \Gamma$, it already appears in~$P_0$.
    Furthermore, $C$ dominates~$g(C)$,
    so we can lower~$C$ to obtain~$g(C)$.
    The subderivation consists of (zero or more) resolutions
    with $2$-clauses to perform the lowering.
  \item[Case 2:] \textit{$C$ is the resolvent of $A$~and~$B$ on~$y$ in~$\Pi$.}
    The subderivation~$P_i$ will lower the earlier-derived clauses
    $g(A)$~and~$g(B)$ to form clauses $A'$~and~$B'$
    that can be resolved to give the clause~$g(C)$.
    For each literal~$p$, define
    \begin{equation*}
      \lambda(p) \coloneqq
      \begin{cases}
        L_C(p) + 1 & \text{if $\var(p) = y$} \\
        L_C(p) & \text{otherwise}.
      \end{cases}
    \end{equation*}
    Then $A'$~and~$B'$ are defined as
    \begin{equation*}
      A' = \bigvee_{p \in A} W[p, \lambda(p)] \quad \text{ and } \quad
      B' = \bigvee_{p \in B} W[p, \lambda(p)].
    \end{equation*}
    Since $\Pi$ contains paths from~$A$ and from~$B$
    that resolve on~$y$ and then pass through~$C$,
    we have $\lambda(p) \leq j$ for all~$W[p, j]$ in $g(A)$~or~$g(B)$.
    Therefore, $A'$~and~$B'$ are indeed dominated by $g(A)$~and~$g(B)$,
    and thus can be derived in~$P_i$.
    After that, $P_i$ resolves $A'$~and~$B'$
    on~$W[y, \lambda(y)]$ to derive~$g(C)$.
    \cref{fig:lowering-resolution-premises} summarizes
    the derivation of~$g(C)$ from $g(A)$~and~$g(B)$.

    \begin{figure}[t]
      \centering
      \begin{tikzpicture}[
        line width=0.175mm,
        horizontal center at origin,
        every node/.style={inner sep=2pt}
        ]
        \begin{scope}[rotate=45]
          \node (gA) at (0, 5.5) {$g(A)$};
          \coordinate (D1') at (0, 4.5);
          \coordinate (D1) at ($(D1') + (0:1)$);
          \coordinate (D2') at (0, 3.5);
          \coordinate (D2) at ($(D2') + (0:1)$);
          \coordinate (s) at ($(D2') + (0, -0.35)$);
          \coordinate (D3') at (0, 2);
          \coordinate (t) at ($(D3') + (0, 0.35)$);
          \coordinate (D3) at ($(D3') + (0:1)$);
          \node (A') at (0, 1) {$A'$};
        \end{scope}

        \draw (gA) edge (D1') (D1') edge (D2') (D3') edge (A');
        \draw (D2') -- (s);
        \draw[dots=10 per 1cm, line width=0.3mm, shorten <=(1cm/10)] (s) -- (t);
        \draw (t) -- (D3');
        \draw (D1') -- (D1);
        \draw (D2') -- (D2);
        \draw (D3') -- (D3);

        \begin{scope}[rotate=-45, xscale=-1]
          \node (gB) at (0, 5.5) {$g(B)$};
          \coordinate (E1') at (0, 4.5);
          \coordinate (E1) at ($(E1') + (0:1)$);
          \coordinate (E2') at (0, 3.5);
          \coordinate (E2) at ($(E2') + (0:1)$);
          \coordinate (u) at ($(E2') + (0, -0.35)$);
          \coordinate (E3') at (0, 2);
          \coordinate (v) at ($(E3') + (0, 0.35)$);
          \coordinate (E3) at ($(E3') + (0:1)$);
          \node (B') at (0, 1) {$B'$};
        \end{scope}

        \draw (gB) edge (E1') (E1') edge (E2') (E3') edge (B');
        \draw (E2') -- (u);
        \draw[dots=10 per 1cm, line width=0.3mm, shorten <=(1cm/10)] (u) -- (v);
        \draw (v) -- (E3');
        \draw (E1') -- (E1);
        \draw (E2') -- (E2);
        \draw (E3') -- (E3);

        \node (gC) at (0, 0) {$g(C)$};

        \draw (B') -- (gC);
        \draw (A') -- (gC);
      \end{tikzpicture}%
      \caption{Lowering the premises for a resolution inference in Case~2.
        The unlabeled leaves correspond to initial clauses
        from~$f(\Gamma, h) \setminus \Gamma$, included earlier in~$P_0$.}
      \label{fig:lowering-resolution-premises}
    \end{figure}
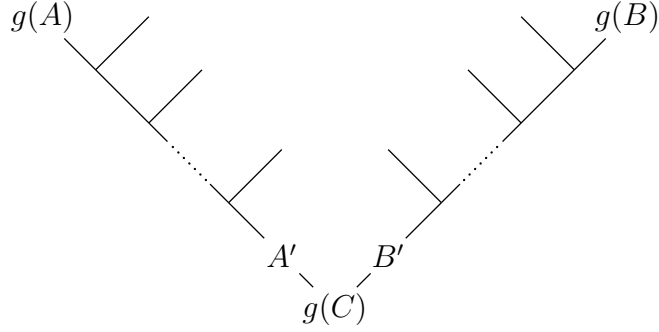
  \end{description}

  The refutation~$\Pi'$ is completed once
  $P_s$ is formed, as it derives $g(C_s) = \bot$.
  By construction, for every~$x$,
  the levels~$j$ of the resolution variables~$W[x, j]$
  are decreasing along all paths in the refutation~$\Pi'$.
  Therefore, $\Pi'$ is a regular resolution refutation.

  It is straightforward to see that $\Pi'$ has size
  at most $s + 2hn + (2hn + 1)s \leq 6hns$.
  This bound is calculated as follows:
  $\Pi'$~has at most~$s$ clauses from~$\Gamma$.
  It also has the $\leq 2hn$~many $2$-clauses added
  in the definition~\cref{eq:addedTwoClauses} of~$f(\Gamma, h)$.
  The term~$(2hn + 1)s$ is justified by noting
  that for each of the $\leq s$~resolution inferences
  in the original refutation~$\Pi$,
  lowering is performed in Case~2 as needed on the literals in the clauses
  (at most $h$~times on each of the $\leq 2n$~literals
  in the clauses being resolved)
  and then one resolution inference is performed
  by resolving on the lowered version of~$y$.

  The bound~$hn$ on the height of~$\Pi'$ follows from the fact
  that it is a regular resolution refutation of a formula with $hn$~variables,
  since each variable can be resolved on at most once
  along any path in the refutation~$\Pi'$.
\end{proof}

The proof of \cref{thm:regular-effectively-general} establishes
a statement stronger than necessary for the effective simulation to hold.
\cref{def:effective-simulation} allows~$f$ to depend on proof size;
in our case it suffices for~$f$ to depend on height.
This dependence is needed in our simulation
to ensure that there are sufficiently many variables~$W[x, j]$.
We leave open whether the dependence can be eliminated:

\begin{question}\label{qn:regular-strictly-effectively-general}
  Does regular resolution strictly effectively simulate resolution?
\end{question}

\section{Corollaries}
\label{sec:corollaries}

\cref{thm:regular-effectively-general} has some interesting consequences,
given in \cref{cor:regular-under-substitutions,%
  cor:regular-weak-automatability,%
  cor:regular-automatability}.
We need a few definitions before we can state the corollaries.

\subsection{Closure under substitutions}
\label{sec:closure-under-substitutions}

Let $\bV$~and~$\bL$ denote the sets of
all variables and all literals, respectively.
A \emph{substitution} is a partial function~%
$\sigma \colon \bV \to \bL \cup \{\false, \true\}$.
We allow~$\sigma$ to act on literals
by letting $\sigma(\lneg{x}) \coloneqq \lneg{\sigma(x)}$.
We call a set \emph{tautological}
if it contains~$\true$ or a pair of complementary literals.
For a clause~$C$, we define
$\sigma(C) \coloneqq \{\sigma(p) \colon p \in C\}$ and
\begin{equation*}
  C|_\sigma \coloneqq
  \begin{cases}
    \true & \text{if } \sigma(C) \text{ is tautological} \\
    \sigma(C) \setminus \{\false\} & \text{otherwise.}
  \end{cases}
\end{equation*}
For a formula~$\Gamma$, we define
$\Gamma|_\sigma \coloneqq \{C|_\sigma \colon
C \in \Gamma \text{ and } C|_\sigma \neq \true\}$.

\begin{definition}\label{def:closure-under-substitutions}
  A proof system~$P$ is \emph{closed under substitutions}
  if for every formula~$\Gamma$ and every substitution~$\sigma$,
  the formula~$\Gamma|_\sigma$ has a~$P$-proof
  of size polynomial in the size of the smallest~$P$-proof of~$\Gamma$.
  We say $P$ is \emph{p-closed under substitutions}
  if there exists an algorithm that,
  given a size-$s$ $P$-proof of~$\Gamma$ and a substitution~$\sigma$,
  outputs a~$P$-proof of~$\Gamma|_\sigma$ in time polynomial in~$s$.
\end{definition}

Most of the natural proof systems are closed under substitutions.
\cref{thm:regular-effectively-general} can be used to prove
that this is not the case for regular resolution;
in fact, due to the exponential separation
between regular resolution and resolution~\cite{AJPU07},
applying substitutions to formulas can increase their
regular resolution refutation complexity exponentially.

\begin{corollary}\label{cor:regular-under-substitutions}
  Regular resolution is not closed under substitutions.
\end{corollary}
\begin{proof}
  Let $\{\Gamma_n\}_{n = 1}^\infty$ be a family of formulas
  admitting polynomial-size refutations in resolution
  while requiring exponential-size refutations in regular resolution.
  For all~$n$, let $h_n$ denote the height of
  the smallest resolution refutation of~$\Gamma_n$.

  Let $f$ be the formula transformation
  defined in the proof of \cref{thm:regular-effectively-general},
  and consider a substitution~$\sigma_n$ that maps~$W[x, j]$ to~$x$
  for all~$x \in \var(\Gamma_n)$ and~$j \in [h_n - 1]$.
  The formula~$f(\Gamma_n, h_n)$ can be
  refuted in polynomial size in regular resolution,
  whereas $f(\Gamma_n, h_n)|_{\sigma_n} = \Gamma_n$ requires exponential size.
\end{proof}

\subsection{Automatability}
\label{sec:automatability}

\begin{definition}\label{def:automatability}
  A proof system~$P$ is \emph{automatable}
  if there exists an algorithm~$A$ that,
  given an unsatisfiable formula~$\Gamma$,
  outputs a~$P$-proof of~$\Gamma$ in time polynomial in~$\sizeOf{\Gamma} + s$,
  where $s$ is the size of the smallest~$P$-proof of~$\Gamma$.
  We say $P$ is \emph{weakly automatable}
  if the algorithm~$A$ is allowed to output a proof in some other system.
\end{definition}

Effective simulations give reductions between
weak automatability of proof systems:

\begin{proposition}[{\cite[Proposition~2.7]{PS10}}]%
  \label{prop:effective-simulations-reduce-weak-automatability}
  Let $P$~and~$Q$ be proof systems.
  If $P$ effectively simulates~$Q$ and $P$ is weakly automatable,
  then $Q$ is weakly automatable.
\end{proposition}

Thus, the following is an immediate corollary of
\cref{thm:regular-effectively-general}.

\begin{corollary}\label{cor:regular-weak-automatability}
  If regular resolution is weakly automatable, then so is resolution.
\end{corollary}

Although effective simulations do not necessarily give
reductions between (strong) automatability,
in our case \cref{cor:regular-weak-automatability} can be extended:

\begin{corollary}\label{cor:regular-automatability}
  If regular resolution is automatable, then so is resolution.
\end{corollary}
\begin{proof}
  Let $A$ be an algorithm that automates regular resolution
  in time bounded by a polynomial~$t$.
  Let $f$ be the formula transformation
  defined in the proof of \cref{thm:regular-effectively-general},
  and let $u$ be a polynomial such that
  when $s$ is at least the size of
  the smallest resolution refutation of~$\Gamma$,
  the formula~$f(\Gamma, s)$ has
  a regular resolution refutation of size~$u(\sizeOf{\Gamma} + s)$.

  Resolution can be automated as follows:
  Given an unsatisfiable formula~$\Gamma$, for each $r = 0, 1, \dots$,
  simulate~$A$ on~$f(\Gamma, r)$ for $t(u(\sizeOf{\Gamma} + r))$~steps
  until, for some~$r$, it outputs
  a regular resolution refutation~$\Pi$ of~$f(\Gamma, r)$.
  Let $\sigma$ be a substitution that maps~$W[x, j]$ to~$x$
  for all~$x \in \var(\Gamma)$ and~$j \in [r - 1]$.
  Resolution is p-closed under substitutions
  and we have $f(\Gamma, r)|_\sigma = \Gamma$, so convert~$\Pi$
  into a resolution refutation~$\Pi'$ of~$\Gamma$ and output~$\Pi'$.
\end{proof}

An analogous result is already known:
Atserias and Müller~\cite{AM20} proved that
resolution is not automatable unless $\PTIME = \NP$,
and it was observed afterwards
that their result can be extended to regular resolution.
However, this extension is a nontrivial step
and requires at least an inspection of their proof.
(See for instance the preprint by Bell~\cite{Bel20} for a detailed writeup.)
In contrast, \cref{cor:regular-automatability} recovers
the same extension in a black-box fashion.

\section*{Acknowledgments}
This work was done in part while the authors were visiting
the Simons Institute for the Theory of Computing.
Sam Buss's research is supported in part by Simons Foundation grant 578919.

\newcommand{\Proc}[1]{Proceedings of the \nth{#1}}
\newcommand{\STOC}[1]{\Proc{#1} Symposium on Theory of Computing (STOC)}
\newcommand{\FOCS}[1]{\Proc{#1} Symposium on Foundations of Computer Science
  (FOCS)} \newcommand{\SODA}[1]{\Proc{#1} Symposium on Discrete Algorithms
  (SODA)} \newcommand{\CCC}[1]{\Proc{#1} Computational Complexity Conference
  (CCC)} \newcommand{\ITCS}[1]{\Proc{#1} Innovations in Theoretical Computer
  Science (ITCS)} \newcommand{\ICS}[1]{\Proc{#1} Innovations in Computer
  Science (ICS)} \newcommand{\ICALP}[1]{\Proc{#1} International Colloquium on
  Automata, Languages, and Programming (ICALP)}
\newcommand{\STACS}[1]{\Proc{#1} Symposium on Theoretical Aspects of Computer
  Science (STACS)} \newcommand{\MFCS}[1]{\Proc{#1} International Symposium on
  Mathematical Foundations of Computer Science (MFCS)}
\newcommand{\LICS}[1]{\Proc{#1} Symposium on Logic in Computer Science
  (LICS)} \newcommand{\CSL}[1]{\Proc{#1} Conference on Computer Science Logic
  (CSL)} \newcommand{\CSLw}[1]{\Proc{#1} International Workshop on Computer
  Science Logic (CSL)} \newcommand{\DLT}[1]{\Proc{#1} International Conference
  on Developments in Language Theory (DLT)} \newcommand{\TAMC}[1]{\Proc{#1}
  International Conference on Theory and Applications of Models of Computation
  (TAMC)} \newcommand{\RTA}[1]{\Proc{#1} International Conference on Rewriting
  Techniques and Applications (RTA)} \newcommand{\IJCAR}[1]{\Proc{#1}
  International Joint Conference on Automated Reasoning (IJCAR)}
\newcommand{\CADE}[1]{\Proc{#1} Conference on Automated Deduction (CADE)}
\newcommand{\SAT}[1]{\Proc{#1} International Conference on Theory and
  Applications of Satisfiability Testing (SAT)}
\newcommand{\TABLEAUX}[1]{\Proc{#1} International Conference on Automated
  Reasoning with Analytic Tableaux and Related Methods (TABLEAUX)}
\newcommand{\TACAS}[1]{\Proc{#1} International Conference on Tools and
  Algorithms for the Construction and Analysis of Systems (TACAS)}
\newcommand{\LPAR}[1]{\Proc{#1} International Conference on Logic for
  Programming, Artificial Intelligence and Reasoning (LPAR)}
\newcommand{\HVC}[1]{\Proc{#1} Haifa Verification Conference (HVC)}
\newcommand{\DAC}[1]{\Proc{#1} Design Automation Conference (DAC)}
\newcommand{\DATE}{Proceedings of the Design, Automation and Test in Europe
  Conference (DATE)} \newcommand{\ISAIM}[1]{\Proc{#1} International Symposium
  on Artificial Intelligence and Mathematics (ISAIM)}
\newcommand{\AAAI}[1]{\Proc{#1} AAAI Conference on Artificial Intelligence
  (AAAI)} \newcommand{\ICML}[1]{\Proc{#1} International Conference on Machine
  Learning (ICML)} \newcommand{\NeurIPS}[1]{\Proc{#1} Conference on Neural
  Information Processing Systems (NeurIPS)}

\end{document}